\documentclass[a4paper, 12pt]{article}

\usepackage[utf8]{inputenc}
\usepackage[T1]{fontenc}
\usepackage[a4paper, margin=2cm]{geometry}
\usepackage{needspace}

\usepackage{libertine} 
\usepackage{inconsolata} 

\usepackage{amsmath, amsthm, amssymb}
\usepackage{graphicx}
\usepackage{enumerate}
\usepackage{authblk}
\usepackage[textsize=small, textwidth=2cm, color=yellow]{todonotes}
\usepackage{thmtools}
\usepackage{thm-restate}
\usepackage[colorlinks=true, citecolor=red]{hyperref}
\usepackage{float}
\usepackage{xparse}
\usepackage{relsize}
\usepackage[qm]{qcircuit}
\usepackage{tikz}
\usetikzlibrary{calc,positioning}

\declaretheorem[name=Theorem, numberwithin=section]{theorem}
\declaretheorem[name=Lemma, sibling=theorem]{lemma}

\declaretheorem[name=Definition, sibling=theorem]{definition}
\declaretheorem[name=Corollary, sibling=theorem]{corollary}
\declaretheorem[name=Conjecture, sibling=theorem]{conjecture}
\declaretheorem[name=Problem, sibling=theorem]{problem}

\def\cqedsymbol{\ifmmode$\lrcorner$\else{\unskip\nobreak\hfil
\penalty50\hskip1em\null\nobreak\hfil$\lrcorner$
\parfillskip=0pt\finalhyphendemerits=0\endgraf}\fi}

\newcommand{\ket}[1]{\left\vert #1 \right\rangle}

\newcommand{\pauli}[1]{\ensuremath{\mathtt{#1}}}
\newcommand{\qop}[1]{\normalfont{\text{\textlarger[1]{\textsc{#1}}}}}
\renewcommand{\S}{\qop{s}}
\renewcommand{\H}{\qop{h}}
\newcommand{\T}{\qop{t}}
\newcommand{\CNOT}{\qop{cnot}}

\newcommand{\CZ}{\qop{cz}}
\newcommand{\SWAP}{\qop{swap}}
\newcommand{\iSWAP}{\qop{{\textsmaller[1]{$i$}}swap}}
\newcommand{\CZdistance}{\CZ-distance}

\DeclareDocumentCommand{\PauliGroup}{m o}  
  {\IfValueTF{#2}{\mathcal{P}^{#2}}{\mathcal{P}}}
\DeclareDocumentCommand{\CliffordGroup}{m o}  
  {\IfValueTF{#2}{\mathcal{C}^{#2}}{\mathcal{C}}}

\definecolor{new_edge}{rgb}{0.1,0.9,0.1}
\definecolor{old_edge}{rgb}{0.85,0.85,0.85}

\let\le\leqslant
\let\ge\geqslant

\title{Preparing graph states forbidding a vertex-minor}

\author[1]{James Davies}
\author[2]{Andrew Jena}

\affil[1]{University of Cambridge, United Kingdom.}
\affil[2]{University of Waterloo, Canada.}
\date{}

\begin{document}

\maketitle

\begin{abstract}
Measurement based quantum computing is preformed by adding non-Clifford measurements to a prepared stabilizer states.
Entangling gates like {\normalfont{\text{\textlarger[1]{\textsc{cz}}}}} are likely to have lower fidelities due to the nature of interacting qubits, so when preparing a stabilizer state, we wish to minimize the number of required entangling states.
This naturally introduces the notion of {\normalfont{\text{\textlarger[1]{\textsc{cz}}}}}-distance.

Every stabilizer state is local-Clifford equivalent to a graph state, so we may focus on graph states $\left\vert G \right\rangle$.
As a lower bound for general graphs, there exist $n$-vertex graphs $G$ such that the {\normalfont{\text{\textlarger[1]{\textsc{cz}}}}}-distance of $\left\vert G \right\rangle$ is $\Omega(n^2 / \log n)$.
We obtain significantly improved bounds when $G$ is contained within certain proper classes of graphs.
For instance, we prove that if $G$ is a $n$-vertex circle graph with clique number $\omega$, then $\left\vert G \right\rangle$ has {\normalfont{\text{\textlarger[1]{\textsc{cz}}}}}-distance at most $4n \log \omega + 7n$.
We prove that if $G$ is an $n$-vertex graph of rank-width at most $k$, then $\left\vert G \right\rangle$ has {\normalfont{\text{\textlarger[1]{\textsc{cz}}}}}-distance at most $(2^{2^{k+1}} + 1) n$.
More generally, this is obtained via a bound of $(k+2)n$ that we prove for graphs of twin-width at most~$k$.

We also study how bounded-rank perturbations and low-rank cuts affect the {\normalfont{\text{\textlarger[1]{\textsc{cz}}}}}-distance.
As a consequence, we prove that Geelen's Weak Structural Conjecture for vertex-minors implies that if $G$ is an $n$-vertex graph contained in some fixed proper vertex-minor-closed class of graphs, then $\left\vert G \right\rangle$ has {\normalfont{\text{\textlarger[1]{\textsc{cz}}}}}-distance at most $O(n\log n)$.
Since graph states of locally equivalent graphs are local Clifford equivalent, proper vertex-minor-closed classes of graphs are natural and very general in this setting.
\end{abstract}

\section{Introduction}\label{sec:intro}

As a consequence of the low coherence times for qubits and low fidelities for quantum gates, the task of minimising the number of gates required to implement a given unitary is ever important. In this paper, we will focus on Clifford unitaries, which, famously according to the Gottesman-Knill Theorem, can be efficiently simulated on a classical device \cite{gottesman1998heisenberg}.

Despite, or perhaps owing to, being efficiently simulatable, Clifford circuits appear in a number of important applications. Adding just one non-Clifford gate to the Clifford gate set results in a university gate set, which can approximate any unitary up to arbitrary precision. One approach to implementing these non-Clifford operations is magic state distillation \cite{bravyi2005universal}, and a popular example of a universal gate set constructed in this way is Clifford+\T. Another approach to achieving universal computation is by adding non-Clifford measurements, as in measurement based quantum computing (MBQC), also known as one-way quantum computing \cite{raussendorf2001one, van2004efficient, browne2016one}.

Clifford circuits and related stabilizer states play a fundamental role in the theory of quantum error correction \cite{gottesman1997stabilizer, gottesman2016surviving, nielsen2001quantum}. Graph states, the representation of stabilizer states which we will focus on here, have been used to understand entanglement in complicated multipartite states thanks to their elucidating combinatorial representation \cite{hein2006entanglement}. As a last example, a key bottleneck in variational quantum algorithms is the number of measurements required by the choice of measurement bases~\cite{mcclean2016theory, cerezo2021variational}, and measuring in distinct Clifford bases has been a productive approach to minimising this cost \cite{huang2020predicting, verteletskyi2020measurement, shlosberg2023adaptive}.

Our basis set will consist of the single-qubit gates, \H{} and \S{}, and the two-qubit entangling gate, \CZ{}, which together generate the $n$-qubit Clifford group up to a global phase which we ignore. An unfortunate reality is that entangling gates like \CZ{} are likely to have lower fidelities due to the nature of interacting qubits, so we will be minimising the number of \CZ{} gates without considering the cost of \H{} and \S{} gates. Up to conjugation by single-qubit gates, \CZ{} and \iSWAP{} are representatives of the only two inequivalent cosets of two-qubit Clifford entangling gates \cite{grier2022classification}. Both \CZ{} and \iSWAP{} can be generated using two applications of the other, and since \CZ{} gates (which are single-qubit equivalent to \CNOT{} gates) are more popular in the literature, we will focus our attention on the basis $\{\H{}, \S{}, \CZ{}\}$, with the caveat that a similar analysis would result in similar proofs for $\{\H{}, \S{}, \iSWAP{}\}$ and $\{\H{}, \S{}, \CZ{}, \iSWAP{}\}$ with at most a factor of 2 difference.

It is not true that every stabilizer state is a graph state, but it \emph{is} true that every stabilizer state is local-Clifford equivalent to a graph state \cite{schlingemann2001stabilizer}. 
The problem of minimising the number of $\CZ$ gates needed to prepare a stabilizer state or graph state is naturally captured by the notion of \CZdistance{}, which was developed by the present authors in the second author's PhD thesis, where it was referred to as ``entanglement distance'' \cite{jena2024graph}.
For a graph state $\ket{G}$, we say that the \CZdistance{} of $\ket{G}$ is equal to the minimum number of $\CZ$ gates required to prepare $\ket{G}$ while only using $\H, \S,$ and $\CZ$ gates.

The action of single-qubit Clifford gates on stabilizer states can be easily translated in the graph state formalism using a graph operation known as a local complementation \cite{van2004graphical}. Let $G$ be a graph and let $v \in V(G)$. We say that the graph $G*v$ obtained from $G$ by complementing the induced subgraph on the neighbours of $v$ is the graph obtained from $G$ by performing a \emph{local complementation} at $v$. Two graphs, $G$ and $H$, are \emph{locally equivalent} if $H$ can be obtained from $G$ by a sequence of local complementations. A graph $H$, is a \emph{vertex-minor} of a graph, $G$, if it is an induced subgraph of some graph locally equivalent to $G$, or equivalently, if it can be obtained from $G$ by a sequence of vertex deletions and local complementations.

As alluded to above, there is a correspondence between local complementations and single-qubit Clifford operations. The graph state $\ket{G*v}$ can be obtained from the graph state $\ket{G}$ by applying a product of single-qubit unitaries
\[
\sqrt{i\pauli{X}_v} \hspace{0.5em} \cdot \hspace{-0.5em} \prod_{u \in N_G(v)} \hspace{-0.75em} \sqrt{i\pauli{Z}_u}
\]
where $\pauli{X}$ and $\pauli{Z}$ are the Pauli X and Z gates, respectively.
$\sqrt{\pauli{X}} = \H{} \S{} \H{}$ and $\sqrt{\pauli{Z}} = \S{}$, so up to a global phase, this can be constructed from single-qubit Clifford gates in our basis.

Similarly, there is a correspondence between adding or deleting an edge $\{uv\}$, and a \CZ{} gate on qubits $u$ and $v$.
The graph state $\ket{G \Delta \{uv\}}$ can be obtained from the graph state $\ket{G}$ by applying the two-qubit entangling gate
\[
\CZ_{u,v}.
\]
For readers who are less familiar with the graph and stabilizer states or readers who have not seen the derivation of these facts, Appendix~\ref{app:quantum} provides quick definitions of Paulis, Cliffords, and graph states.
Verifying the above correspondences directly from these definitions is a good exercise.

If $G$ is a connected $n$-vertex graph, then $\ket{G}$ has \CZdistance{} at least $n-1$.
The \CZdistance{} of arbitrary $n$-vertex graphs is well understood; the maximum \CZdistance{} of an $n$-vertex graph is $\Theta(n^2/\log n)$ \cite{patel2008optimal}.
So to achieve better upper bounds than $O(n^2 / \log n)$, we must restrict to proper classes of graphs.
Since locally equivalent graph states can be obtained from each other via single-qubit Clifford gates, it is natural to consider graph classes that are closed under local complementations.
From this point of view, the most general proper classes of graphs in our setting are then classes of graphs forbidding a vertex-minor.
Proper vertex-minor-closed classes include several natural classes of graphs such as circle graphs and graphs of bounded rank-width.
For a recent survey on vertex-minors, see \cite{kim2023vertex}.

Geelen \cite{mccarty2021local} has a widely believed structural conjecture for vertex-minor-closed classes (see Conjecture \ref{conject:weakstructure} for the weak version), which there has been significant progress towards proving (see~\cite{geelen2023grid,mccarty2021local,mccarty2023decomposing}), and we are optimistic that it will be proven in the coming years.
Our main result is that, assuming the weak version of Geelen's vertex-minor structure theorem (see Conjecture \ref{conject:weakstructure}), we have a near linear improvement from $O(n^2/\log n)$ to $O(n\log n)$ for the \CZdistance{} of a graph state $\ket{G}$ when $G$ is a $n$-vertex graph in a proper vertex-minor-closed class of graphs.

\begin{theorem}\label{CZ:mainvertexminor}
Conjecture \ref{conject:weakstructure} implies the following.
    Let $\mathcal{F}$ be a proper vertex-minor-closed class of graphs, and let $G$ be an $n$-vertex graph contained in $\mathcal{F}$.
    Then $\ket{G}$ has \CZdistance{} at most $O(n \log n)$.    
\end{theorem}

Circle graphs and graphs of bounded rank-width are the most fundamental vertex-minor-closed class of graphs. Circle graphs are believed to play a role for vertex-minors that is analogous to that of planar graphs for graph minors, while rank-width is believed to play a role analogous to tree-width (see \cite{geelen2023grid,mccarty2021local}). One striking result is an analogue of Kuratowski's theorem characterising circle graphs by a list of three forbidden vertex-minors \cite{bouchet1994circle}.
We also obtain improved explicit bounds for these two most natural classes of vertex-minor closed classes.
The \emph{clique number} of a graph is equal the the size of its largest complete subgraph.

\begin{theorem}\label{CZ:circleclique}
    Let $G$ be an $n$-vertex circle graph with clique number at most $\omega$.
    Then the \CZdistance{} of $\ket{G}$ is at most $4 n\log \omega + 7n$.
\end{theorem}

\begin{theorem}\label{CZ:rankmain}
    Let $G$ be an $n$-vertex graph of rank-width at most $k$.
    Then the \CZdistance{} of $\ket{G}$ is at most $(2^{2^{k+1}}  +1)n$.
\end{theorem}

We actually obtain Theorem \ref{CZ:rankmain} from a stronger theorem on graphs with bounded twin-width, which was recently introduced by Bonnet, Kim, Thomass{\'e}, and Watrigant~\cite{bonnet2021twin}.

\begin{theorem}\label{CZ:twin}
    Let $G$ be an $n$-vertex graph of twin-width at most $k$.
    Then the \CZdistance{} of $\ket{G}$ is at most $(k  +2)n$.
\end{theorem}

Other classes of graphs that have bounded twin-width include proper hereditary subclasses of permutation graphs, proper minor-closed classes of graphs, graphs of bounded stack or queue number, and $k$-planar graphs \cite{bonnet2021twin}.

Kumabe, Mori, and Yoshimura \cite{kumabe2024complexity} recently concurrently and independently introduced the related notion of \CZ{}-complexity and also studied it for classes such as circle graphs
and of graphs of bounded rank-width.
The notion of \CZ{}-complexity is similar to \CZdistance{} except in that it allows arbitrarily many measurements in a Clifford basis and that it allows arbitrarily many \SWAP{} operations.
So, the \CZ{}-complexity of $G$ is the minimum \CZdistance{} of $H$ over all $H$ for which $G$ is a vertex-minor up to a constant factor of 2.
We conjecture that these definitions differ by no more than a constant factor, but this remains an open problem.
Since optimizations of \CZ \ counts will prove most critical on near- and medium-term quantum devices, where ancillary qubits may be prohibitively expensive, we prefer to not assume access to arbitrarily many qubits.
Kumabe, Mori, and Yoshimura \cite{kumabe2024complexity} proved that circle graphs have \CZ{}-complexity $O(n \log n)$ and that graphs of rank-width at most $k$ have \CZ{}-complexity $O(kn)$.

In Section~\ref{sec:circle}, we introduce some basic further definitions and preliminaries.
In Section~\ref{sec:circle}, we prove that circle graphs have \CZdistance{} $O(n \log n)$. In Section~\ref{sec:perturbations}, we prove that rank-$p$ perturbations only change the \CZdistance{} by $O(pn)$. In Section~\ref{sec:vertex_minors}, we prove that, assuming Geelen's Weak Structural Conjecture holds, these results would imply that the \CZdistance{} of any graph in a proper vertex-minor-closed class of graphs is $O(n \log n)$. In Section~\ref{sec:twin_width}, we prove that graphs with bounded rank width have \CZdistance{} $O(n)$ and therefore prove a logarithmic improvement for vertex-minor-closed classes of graphs not containing all circle graphs. We conclude in Section~\ref{sec:conclude} by conjecturing whether the $O(n \log n)$ bound on circle graphs can be improved and discuss future directions.

\section{Preliminaries}\label{sec:pre}

It is convenient to define \CZdistance{} not just for individual graphs, but also for between two graphs.

\begin{definition}[\CZdistance{}]
Given two graphs, $G$ and $H$, on the same vertex set, we define the \emph{\CZdistance{}} between $G$ and $H$ (denoted $\CZ{}(G,H)$) to be equal to the minimum $k$ such that there is a
sequence of graphs
\[
G = G_0, G_0', G_1, G_1', \dotsc, G_k, G_k' = H
\]
satisfying the following conditions:
\begin{itemize}
\item for each $0 \le i \le k$, $G_i'$ is locally equivalent to $G_i$, and
\item for each $1 \le i \le k$, $G_i$ is obtained from $G_{i-1}'$ by adding or removing a
single edge.
\end{itemize}
\end{definition}
For a single graph, $G$, we let $\CZ{}(G)$ be equal to the \CZdistance{} between $G$ and the edge-less graph on the same vertex set.
This is clearly equivalent to our previous definition of \CZdistance{}.
Notice that this satisfies the triangle inequality: for three graphs, $G, H, F$, we have that $\CZ{}(G,H) \le \CZ{}(G,F) + \CZ{}(F,H)$. In particular, for a pair of graphs, $G$ and $H$, we have $\CZ{}(G,H) \le \CZ{}(G) + \CZ{}(H)$.

For a $n$-vertex graph $G$ and some $F\subseteq E (K_n)$ (where $K_n$ has the same vertex set as $G$), we let $G\Delta F$ denote the graph obtained from $G$ by complementing the edges $F$.
For a vertex $u$ of a graph $G$, we denote by $N_G(u)$ the neighborhood of $u$, or the vertices of $G$ adjacent to $u$.

We will often use the following key observation which can easily be verified.

\begin{lemma}\label{key}
    Let $u,v$ be distinct vertices of a graph $G$, and let $H=(((G*v) \Delta \{uv\})*v) \Delta \{uv\}$, then $H$ is the graph obtained from $G$ by removing all edges between $u$ and $N_G(u) \cap N_G(v)$ and adding all edges between $u$ and $N_G(v)\backslash (N_G(u)\cup \{u\})$.
    As a consequence, $\CZ{}(G,H)\le 2$.
\end{lemma}

\section{Circle graphs}\label{sec:circle}

For a graph $G$ and two disjoint vertex sets $A,B\subseteq V(G)$, we let $G[A,B]$ be the bipartite subgraph of $G$ on vertex set $A\cup B$ where $xy$ is an edge of $G[A,B]$ if and only if $xy$ is an edge of $G$, and $x\in A$, $y\in B$.

Two intervals $I_1,I_2$ in $\mathbb{R}$ \emph{overlap} if they intersect and neither is contained in the other.
For a collection of closed intervals $\mathcal{I}$ in $\mathbb{R}$, the \emph{overlap graph} $G(\mathcal{I})$ is the graph with vertex set $\mathcal{I}$ and edge set being the pairs of overlapping intervals in $\mathcal{I}$. It is well known (see for example \cite{davies2021circle}) that every circle graph is an overlap graph of a collection of intervals in $\mathbb{R}$ such that no two share an endpoint.

\begin{lemma}\label{circlebipartite}
    Let $G$ be a $n$-vertex circle graph with an isolated vertex $u$, let $A,B\subseteq V(G)\backslash \{u\}$ be disjoint, and let $F$ be the edges between $A$ and $B$ in $G$.
    Let $H$ be a graph on the same vertex set as $G$, with $u$ an isolated vertex. 
    Then $\CZ{}(H, H\Delta F)\le 2n-2$.
\end{lemma}

\begin{proof}
    Let $\mathcal{I}$ be a collection of closed intervals in $\mathbb{R}$ with no two sharing an endpoint such that $G\backslash \{u\}=G(\mathcal{I})$.
    Let $b_1< \cdots < b_{2n-2}$ be the endpoints of $\mathcal{I}$, and for each $b_i$, let $I(b_i)$ be the interval of $\mathcal{I}$ that has $b_i$ as an endpoint.
    For $I\in \mathcal{I}$, let $\ell(I)$ be such that $b_{\ell(I)}$ is the left endpoint of $I$, and let $r(I)$ be such that $b_{r(I)}$ is the right endpoint of $I$.

    Let $H_0=H$.
    Now, for each $1\le i \le 2n-2$ in order,
    \begin{itemize}
        \item if $I(b_i)\in A$, then let $H_{i}=((H_{i-1} * u)\Delta \{uI(b_i)\}) * u$,
        \item if $I(b_i)\in B$, then let $H_i=H_{i-1} \Delta \{uI(b_i)\}$, and
        \item otherwise, let $H_i=H_{i-1}$.
    \end{itemize}
    Clearly for each $i$, we have that $\CZ{}(H_i,H_{i-1})\le 1$, so $\CZ{}(H_{2n-2},H)\le 2n-2$. It remains to show that $H_{2n-2}=H \Delta F$.
    Observe that for each $0\le i \le 2n-2$ in order, we have that $N_{H_i}(u)=\{I\in A \cup B : b_i \in I\}$.
    So, $N_{H_{2n-2}}(u)=N_{H\Delta F}(u)$.
    
    By Lemma \ref{key}, we have that if $I(b_i)\in A$, then $H_i\backslash \{u\} = (H_{i-1} \Delta \{I(b_i)x : x\in N_{H_{i-1}}(u)\} ) \backslash \{u\}$.
    Clearly if $I(b_i) \in B$, then $H_i\backslash \{u\}=H_{i-1} \backslash \{u\}$, and otherwise if $I(b_i) \not\in A \cup B$, then $H_i=H_{i-1}$.
    It now follows that 
    \begin{align*}
    H_{2n-2}
    &=
    H \Delta_{I\in A} 
    \left(  
    \{Ix : x\in N_{H_{\ell(I)-1}}(u)\} 
    \Delta
    \{Ix : x\in N_{H_{r(I)-1}}(u)\}  
    \right)
    \\
    &=
    H \Delta_{I\in A} 
    \left(  
    \{IJ : b_{\ell(I)} J\in A\cup B\} 
    \Delta
    \{IJ : b_{r(I)} J\in A\cup B\}
    \right)
    \\
    &=
    H \Delta_{I\in A}  
    \{IJ: J\in A\cup B, IJ \in G(\mathcal{I}) \} 
    \\
    &=
    H \Delta_{I\in A}  
    \{IJ: J\in  B, IJ \in G(\mathcal{I}) \} 
    \\
    &=
    H \Delta F,
    \end{align*}
as desired.
\end{proof}

With a divide and concur strategy, we can extend this further to circle graphs with bounded chromatic number.
A graph is \emph{$k$-colourable} if there is an assignment of at most $k$ colours to its vertices so that no two adjacent vertices are assigned the same colour. The \emph{chromatic number} $\chi(G)$ of a graph $G$ is equal to the minimum $k$ such that $G$ is $k$-colourable.

\begin{lemma}\label{circlechromaticlem}
    Let $G$ be an $n$-vertex circle graph with chromatic number at most $k$ and an isolated vertex $u$. Then $\CZ{}(G)\le (2n-2) \lceil \log k \rceil$.
\end{lemma}

\begin{proof}
    If $k=1$, then the result is trivial, so we proceed inductively.
    Let $A,B\subseteq V(G)\backslash \{u\}$ be a partition such that $\chi(G[A])= \lfloor \frac{k}{2} \rfloor$ and $\chi(G[B])= \lceil \frac{k}{2} \rceil$. Then, by the inductive hypothesis, we have that
    \begin{align*}
        \CZ{}(G[A] \cup G[B] \cup G[\{u\}]) &\le \CZ{}(G[A\cup \{u\}]) + \CZ{}(G[B\cup \{u\}]) \\
        &\le 2|A|\lceil \log \lfloor \frac{k}{2} \rfloor \rceil + 2|B| \lceil \log \lceil \frac{k}{2} \rceil \rceil \\
        &= (2n-2)(\lceil \log k \rceil -1).
    \end{align*}
    By Lemma \ref{circlebipartite}, we have that $\CZ{}(G, G[A] \cup G[B] \cup G[\{u\}]) \le 2n-2$.
    Therefore, 
    \begin{align*}
        \CZ{}(G) &\le \CZ{}(G[A] \cup G[B] \cup G[\{u\}]) + \CZ{}(G, G[A] \cup G[B] \cup G[\{u\}]) \le (2n-2) \lceil \log k \rceil.
    \end{align*}
\end{proof}

We obtain the following corollary.

\begin{corollary}\label{circlechromaticcol}
    Let $G$ be an $n$-vertex circle graph with chromatic number at most $k$. Then $\CZ{}(G)\le (2n-2) \lceil \log k \rceil + n -1$.
\end{corollary}

\begin{proof}
    Let $u$ be a vertex of $G$ and let $E_u$ be the edges of $G$ incident to $u$. By Lemma \ref{circlechromaticlem}, we have that $\CZ{}(G\backslash E_u) \le (2n-2) \lceil \log k \rceil$.
    Clearly $\CZ{}(G\backslash E_u, G) \le |E_u| \le  n-1$.
    Therefore, 
    \begin{align*}
        \CZ{}(G) \le \CZ{}(G\backslash E_u) + \CZ{}(G\backslash E_u, G) \le (2n-2) \lceil \log k \rceil + n -1.\,
    \end{align*}
    as desired.
\end{proof}

We can improve Corollary \ref{circlechromaticcol} by replacing the dependence on the chromatic number of our circle graph $G$, by its clique number instead.
For this we require the following theorem of Davies and McCarty \cite{davies2021circle} on colouring circle graphs with bounded clique number.

\begin{theorem}[{\cite[Theorem~1]{davies2021circle}}]\label{chi-boundcircle}
    Every circle graph with clique number at most $\omega$ is $7\omega^2$-colourable.
\end{theorem}

From Corollary \ref{circlechromaticcol} and Theorem \ref{chi-boundcircle}, we therefore obtain the following (which is equivalent to Theorem \ref{CZ:circleclique}).

\begin{theorem}\label{circleclique}
    Let $G$ be an $n$-vertex circle graph with clique number at most $\omega$.
    Then $\CZ{}(G) \le  4 n\log \omega + 7n$.
\end{theorem}

A $O(\omega \log \omega)$ bound for the chromatic number of a circle graph is also proven in \cite{davies2022improved}. One can use this to slightly improve Theorem \ref{circleclique} for larger $\omega$.

\section{Perturbations}\label{sec:perturbations}

A \emph{rank-$p$ perturbation} of a graph $G$ is a graph whose
adjacency matrix can be obtained from the adjancency matrix of $G$ by first adding (over the binary field) a symmetric matrix of rank at most $p$, and then changing all diagonal entries to be $0$.
For a graph $G$ and $X\subseteq V(G)$, we say that \emph{complementing on $X$} is the act of obtaining a new graph $H$ from $G$ by replacing the induced subgraph of $G$ on $X$ by its complement.



Lempel's optimal algorithm for matrix factorization over $GF(2)$ implies the following upper bound.
See also \cite{mccarty2021local,nguyen2020average} for proofs with slightly weaker bounds.

\begin{lemma}[{\cite[Theorem~1]{lempel1975matrix}}]\label{perturb-bound}
    Let $G$ be an $n$-vertex graph and let $H$ be a rank-$p$ perturbation of $G$.
    Then $G$ can be obtained from $H$ by complementing on at most $p+1$ sets of vertices.
\end{lemma}

Let us first examine the \CZdistance{} between a graph and another graph obtained by complementing on a set of vertices.

\begin{lemma}\label{pertub-one}
    Let $G$ be an $n$-vertex graph, and let $H$ be obtained by complementing on a set $X\subseteq V(G)$. Then $\CZ{}(G,H)\le 2n-2$.
\end{lemma}

\begin{proof}
    Let $u$ be a vertex of $G$, and let $G_1$ be the graph obtained from $G$ by changing the neighbourhood of $u$ to be $X \backslash \{u\}$. Let $G_2=G_1*u$. Then $H$ is obtained from $G_2$ by simply changing the neighbourhood of $u$.
    Therefore,
    \begin{align*}
        \CZ{}(G,H)\le \CZ{}(G,G_1) + \CZ{}(G_1,G_2) + \CZ{}(G_2, H)\le (n-1) + 0 + (n-1) = 2n-2,
    \end{align*}
    as desired.
\end{proof}

By applying Lemma \ref{perturb-bound} and then repeatedly applying Lemma \ref{pertub-one}, we now obtain the following.

\begin{theorem}\label{main-perturb}
    Let $G$ be an $n$-vertex graph and let $H$ be a rank-$p$ perturbation of $G$.
    Then $\CZ{}(G,H)\le (2p+2)n-2p-2$.
\end{theorem}

An immediate corollary of Theorem \ref{circleclique} and Theorem \ref{main-perturb} is the following.

\begin{corollary}\label{circle-perturb}
    Let $G$ be an $n$-vertex rank-$p$ perturbation of a circle graph. Then $\CZ{}(G) \le 4n \log n + (2p+9)n$.
\end{corollary}

\section{Vertex-minors}\label{sec:vertex_minors}

The \emph{cut-rank} of a set $X \subseteq V(G)$, denoted $\rho(X)$, is the rank of the submatrix of the adjacency matrix with rows $X$ and columns $V(G) - X$.
For $k \in \mathbb{N}$, a graph G is \emph{$k$-rank-connected} if it has at least $2k$ vertices and $\rho(X) \ge \min(|X|, |V(G) - X|, k)$ for each $X \subseteq V(G)$.

Geelen's (see \cite{mccarty2021local}) weak vertex-minor structure conjecture states that every graph in a proper vertex-minor-closed class of graph with sufficiently high rank-connectivity is a bounded rank perturbation of a circle graph.

\begin{conjecture}[Weak Structural Conjecture \cite{mccarty2021local}]\label{conject:weakstructure}
    For any proper vertex-minor-closed class of graphs $\mathcal{F}$, there exist $k, p \in \mathbb{N}$ so that each $k$-rank-connected graph in $\mathcal{F}$ is a rank-$p$ perturbation of a circle graph.
\end{conjecture}

There is also a stronger vertex-minor structure conjecture (see \cite{mccarty2021local}) which handle the case of low rank-connectivity, however we shall not need this for our purposes.
Assuming Conjecture \ref{conject:weakstructure} and by using Corollary \ref{circle-perturb}, we can now derive a $O(n\log n)$ bound for graphs forbidding a vertex-minor (this is equivalent to Theorem \ref{CZ:mainvertexminor}).

\begin{theorem}\label{main-vertex}
    Conjecture \ref{conject:weakstructure} implies the following.
    Let $\mathcal{F}$ be a proper vertex-minor-closed class of graphs, and let $G$ be an $n$-vertex graph contained in $\mathcal{F}$.
    Then $\CZ{}(G)=O(n \log n)$.
\end{theorem}

\begin{proof}
    Let $k,p\in \mathbb{N}$ be as in Conjecture \ref{conject:weakstructure}.
    We shall argue inductively on $n$ that $\CZ{}(G) \le 15k^2pn\log n$.
    We do not attempt to optimize the dependence on $k$ and $p$, since they will likely be huge given a proof of Conjecture \ref{conject:weakstructure}.
    
    If $G$ is a rank-$p$ perturbation of a circle graph, then by Corollary~\ref{circle-perturb},
    \begin{align*}
        \CZ{}(G) &\le 4n \log n + (2p+9)n \le 15pn\log n \le 15k^2pn\log n.
    \end{align*}
    So we may assume that $G$ is not a rank-$p$ perturbation of a circle graph. Therefore, $G$ is not $k$-rank-connected.
    
    If $n\le 2k$, then clearly
    \begin{align*}
        \CZ{}(G) \le |E (G)| \le {n \choose 2} \le {2k \choose 2} \le 2k^2 \le 15k^2pn\log n,
    \end{align*}
    so we may assume that $n>2k$.
    So, there exists some $X\subset V(G)$ with $\rho(X) < \min(|X|, |V(G) - X|, k)$.
    We may choose such a $X$ so that $|X|\ge |V(G)\backslash X|$.
    Let $a_1,\ldots,a_{\rho_G{(X)}}$ be vertices of $V(G)\backslash X$ such that $\rho_G(X) = \rho_G(V(G)\backslash X) = \rho_{G[X\cup A]}(A)$, where $A = \{a_1,\ldots,a_{\rho_G{(X)}}\}$. Let $Y = V(G) \backslash (X \cup A)$.
    Then by the inductive hypothesis, $\CZ{}(G[X\cup A])\le 15k^2p|X\cup A|\log |X\cup A|$, and $\CZ{}(G[Y])\le 15k^2p|Y|\log |Y|$.
    Since $|Y|\le n/2$, we therefore get that
    \begin{align*}
    \CZ{}(G[X\cup A] \cup G[Y]) &\le \CZ{}(G[X\cup A]) + \CZ{}(G[Y]) \\
    &\le 15k^2p|X\cup A|\log |X\cup A| + 15k^2p|Y|\log |Y| \\
    &\le 15k^2p|X\cup A|\log n + 15k^2p|Y|\log \frac{n}{2} \\
    &= 15k^2p|X\cup A|\log n + 15k^2p|Y|(\log n - 1) \\
    &= 15k^2pn\log n - 15k^2p|Y|.
    \end{align*}
    
    Let $G_0$ be the graph obtained from $G[X\cup A] \cup G[Y]$ by removing all edges between the vertices of $A$.
    Then, $\CZ{}(G_0) \le \CZ{}(G[X\cup A] \cup G[Y]) + \frac{1}{2}k(k-1) \le 15k^2pn\log n - 15k^2p|Y| + \frac{1}{2}k^2$ since $|A|=\rho_G{(X)} \le k$.

    Let $G_1$ be the graph obtained from $G_0$ by adding edges between $X$ and $Y$ so that $E_{G_1}(X,Y)=E_G(X,Y)$.
    For each vertex $y\in Y$, there exists some $A_y\subseteq A$ such that $N_G(b)\cap X = \Delta_{a\in A_y} N_{G_0}(a)$. Note that $|A_y|\le |A| \le k$ for each $y\in Y$.
    So, by repeatedly applying Lemma \ref{key} a total of $|A_y|$ times for each vertex $y\in Y$, we have that
    \[
    \CZ{}(G_0,G_1)\le \sum_{y\in Y} 2|A_y| \le 2k|Y|.
    \]
    Therefore $\CZ{}(G_1)\le 15k^2pn\log n - 15k^2p|Y| + \frac{1}{2}k^2 + 2k|Y| \le 15k^2pn\log n - k^2 - k|Y|$.

    Now, $G_1$ and $G$ differ only on the edges with one end in $A$ and the other in $A\cup Y$. So, $\CZ{}(G_1,G)\le |A||A\cup Y|\le k^2 + k |Y|$ since $|A|=\rho_G{(X)} \le k$.
    Therefore
    \[
    \CZ{}(G) \le \CZ{}(G_1) + \CZ{}(G_1,G) \le 15k^2pn\log n
    \]
    as desired.
\end{proof}

\section{Twin-width}\label{sec:twin_width}

In this section, we prove a linear bound for the \CZdistance{} of an $n$-vertex graph of bounded rank-width.
For the proof, it is more convenient to work with graph of bounded twin-width, which was recently introduced by Bonnet, Kim, Thomass{\'e}, and Watrigant \cite{bonnet2021twin}. We can do this instead since graphs of bounded rank-width have bounded twin-width.

\begin{theorem}[{\cite[Theorem~4.2]{bonnet2021twin}}]\label{twinrank}
    Every graph with rank-width at most $k$ has twin-width at most $2^{2^{k+1}}  -1$.
\end{theorem}

Next, we shall formally define twin-width.
A \emph{trigraph} $G$ has vertex set $V(G)$, (black) edge set $E (G)$, and red edge set $R(G)$ (the error
edges), with $E (G)$ and $R(G)$ being disjoint. The set of neighbors $N_G(v)$ of a vertex $v$ in a
trigraph $G$ consists of all the vertices adjacent to $v$ by a black or red edge. A $k$-trigraph is a
trigraph $G$ such that the red graph $(V (G), R(G))$ has degree at most $k$.
A (vertex) contraction or identification in a
trigraph $G$ consists of merging two (non-necessarily adjacent) vertices $u$ and $v$ into a single
vertex $z$, and updating the edges of $G$ in the following way. Every vertex of the symmetric
difference $N_G(u) \Delta N_G(v)$ is linked to $z$ by a red edge. Every vertex $x$ of the intersection
$N_G(u) \cap N_G(v)$ is linked to $z$ by a black edge if both $ux \in E (G)$ and $vx \in E (G)$, and by a red
edge otherwise. The rest of the edges (not incident to $u$ or $v$) remain unchanged. We insist
that the vertices $u$ and $v$ (together with the edges incident to these vertices) are removed
from the trigraph.

A $k$-sequence (or contraction sequence) is a sequence of $k$-trigraphs $G_n, G_{n-1}, \ldots , G_1$,
where $G_n = G$, $G_1 = K_1$ is the graph on a single vertex, and $G_{i-1}$ is obtained from $G_i$ by
performing a single contraction of two (non-necessarily adjacent) vertices. We observe that
$G_i$ has precisely $i$ vertices, for every $1\le i \le n$. The twin-width of $G$, is
the minimum integer $k$ such that G admits a $k$-sequence.

The following lemma is immediate from the definition of twin-width.

\begin{lemma}\label{twindef}
    Let $G$ be an $n$-vertex graph of twin-width at most $k$.
    Then, there exists a sequence of $n$-vertex graphs $G_1,\ldots, G_n$ such that $G_n=G$, $G_1$ is edgeless, and for each $2\le i \le n$, the graph $G_i$ is obtained from $G_{i-1}$ by choosing two vertices $u,v$ of $G_{i-1}$, with $u$ isolated, and then adding all edges between $u$ and $N_{G_{i-1}}(v)$, then possibly adding an edge between $u$ and $v$, and then adding at most $k$ additional edges incident to $u$.
\end{lemma}

The next lemma allows us to bound the \CZdistance{} between $G_{i-1}$ and $G_i$ in the above lemma.

\begin{lemma}\label{twinstep}
    Let $u,v$ be distinct vertices of a graph $G$ with $u$ isolated, and let $H$ be the graph obtained from $G$ by adding all edges between $u$ and $N_G(v)$, then possibly adding an edge between $u$ and $v$, and then adding at most $k$ additional edges incident to $u$.
    Then $\CZ{}(G,H)\le k + 2$.
\end{lemma}

\begin{proof}
    If an edge is added between $u$ and $v$, then observe that $H$ is obtained from $((G*v)\Delta \{uv\})*v$ by adding at most $k$ additional edges, and therefore $\CZ{}(G,H)\le k+1$.
    Otherwise, if no edhe is added between $u$ and $v$, then observe that $H$ is obtained from $((G*v)\Delta \{uv\})*v)\Delta \{uv\}$ by adding at most $k$ additional edges, and therefore $\CZ{}(G,H)\le k+2$.
\end{proof}

We now obtain the following by applying Lemma \ref{twindef} once and then Lemma \ref{twinstep} a total of $n-1$ times (this is equivalent to Theorem \ref{CZ:twin}).

\begin{theorem}\label{twinmain}
    Let $G$ be an $n$-vertex graph of twin-width at most $k$.
    Then $\CZ{}(G)\le (k+2)n$.
\end{theorem}

By Theorem \ref{twinrank} and Theorem \ref{twinmain}, we obtain the following bound for graphs of bounded rank-width (this is equivalent to Theorem \ref{CZ:rankmain}).

\begin{theorem}\label{rankmain}
    Let $G$ be an $n$-vertex graph of rank-width at most $k$.
    Then $\CZ{}(G) \le (2^{2^{k+1}}  +1)n$.
\end{theorem}

The linear bound for graphs of bounded rank-width allows us to obtain a linear bound for proper vertex-minor-closed classes of graphs that do not contain the class of circle graphs. For this we require the vertex-minor grid theorem of Geelen, Kwon, McCarty, and Wollan \cite{geelen2023grid}.

\begin{theorem}[{\cite[Theorem~1]{geelen2023grid}}]\label{grid}
    Let $H$ be a circle graph. Then the class of graphs not containing $H$ as a vertex-minor has bounded rank-width.
\end{theorem}

Now, as an immediate corollary of Theorem \ref{rankmain} and Theorem \ref{grid}, we obtain the following.

\begin{corollary}\label{gridcol}
    Let $\mathcal{F}$ be a proper vertex-minor-closed class of graphs not containing all circle graphs, and let $G$ be an $n$-vertex graph contained in $\mathcal{F}$.
    Then $\CZ{}(G)=O(n)$.
\end{corollary}

\section{Concluding remarks and open problems}\label{sec:conclude}

In this paper we have obtained $O(n\log n)$ bounds for the \CZdistance{} of circle graphs and more generally (assuming Geelen's weak vertex-minor structure theorem) for any proper vertex-minor-closed class of graphs. This is a near linear improvement on the best possible bound of $\Theta(n^2/\log n)$ for the class of all graphs.

Trivially we have the lower bound of $\CZ{}(G)\ge n-1$ for any $n$-vertex connected graph $G$. A small logarithmic gap still remains, leading to the following problem.

    \begin{problem}
        Let $\mathcal{F}$ be a proper vertex-minor-closed class of graphs and let $G$ be an $n$-vertex graph contained in $\mathcal{F}$. Is it true that $\CZ{}(G)=O(n)$?
    \end{problem}

Corollary \ref{gridcol} shows that this is the case when $\mathcal{F}$ does not contain all circle graphs. We believe that this is not the case for circle graphs.

\begin{conjecture}
    There are $n$-vertex circle graphs $G$ with $\CZ{}(G) = \Omega(n \log n)$.
\end{conjecture}

For circle graphs with bounded clique number we were able to obtain a $O(n)$ bound improving the $O(n\log n)$ bound for circle graphs with unbounded clique number. We did this using the fact that circle graphs with bounded clique number have bounded chromatic number. It is also the case that in a proper vertex-minor-closed class, the graphs with bounded clique number have bounded chromatic number \cite{davies2022vertex}. This leads us to the following conjecture.

\begin{conjecture}
    Let $\mathcal{F}$ be a proper vertex-minor-closed class of graphs and let $G$ be an $n$-vertex graph contained in $\mathcal{F}$ with clique number at most $\omega$. Then $\CZ{}(G)=O_\omega(n)$.
\end{conjecture}

Further support for this conjecture is a theorem of Hlin{\v{e}}n{\`y} and Pokr{\`y}vka \cite{hlinveny2022twin} that (assuming Conjecture \ref{conject:weakstructure}), graphs in a proper vertex-minor-closed class of graphs with high rank-connectivity and bounded clique number, have bounded twin-width.

\section*{Acknowledgements}

We thank Luca Dellantonio, Jim Geelen, Lane Gunderman, Rose McCarty for helpful discussions, and Ryuhei Mori for some further helpful comments on the paper.

\bibliographystyle{plain}
\bibliography{quantum}

\begin{thebibliography}{10}

\bibitem{aaronson2004improved}
Scott Aaronson and Daniel Gottesman.
\newblock Improved simulation of stabilizer circuits.
\newblock {\em Physical Review A—Atomic, Molecular, and Optical Physics},
  70(5):052328, 2004.

\bibitem{bonnet2021twin}
{\'E}douard Bonnet, Eun~Jung Kim, St{\'e}phan Thomass{\'e}, and R{\'e}mi
  Watrigant.
\newblock Twin-width {I}: tractable {F}{O} model checking.
\newblock {\em ACM Journal of the ACM (JACM)}, 69(1):1--46, 2021.

\bibitem{bouchet1994circle}
Andr{\'e} Bouchet.
\newblock Circle graph obstructions.
\newblock {\em Journal of Combinatorial Theory, Series B}, 60(1):107--144,
  1994.

\bibitem{bravyi2005universal}
Sergey Bravyi and Alexei Kitaev.
\newblock Universal quantum computation with ideal clifford gates and noisy
  ancillas.
\newblock {\em Physical Review A—Atomic, Molecular, and Optical Physics},
  71(2):022316, 2005.

\bibitem{browne2016one}
Dan Browne and Hans Briegel.
\newblock One-way quantum computation.
\newblock {\em Quantum information: From foundations to quantum technology
  applications}, pages 449--473, 2016.

\bibitem{cerezo2021variational}
Marco Cerezo, Andrew Arrasmith, Ryan Babbush, Simon~C Benjamin, Suguru Endo,
  Keisuke Fujii, Jarrod~R McClean, Kosuke Mitarai, Xiao Yuan, Lukasz Cincio,
  et~al.
\newblock Variational quantum algorithms.
\newblock {\em Nature Reviews Physics}, 3(9):625--644, 2021.

\bibitem{davies2022improved}
James Davies.
\newblock Improved bounds for colouring circle graphs.
\newblock {\em Proceedings of the American Mathematical Society},
  150(12):5121--5135, 2022.

\bibitem{davies2022vertex}
James Davies.
\newblock Vertex-minor-closed classes are $\chi$-bounded.
\newblock {\em Combinatorica}, 42(Suppl 1):1049--1079, 2022.

\bibitem{davies2021circle}
James Davies and Rose McCarty.
\newblock Circle graphs are quadratically $\chi$-bounded.
\newblock {\em Bulletin of the London Mathematical Society}, 53(3):673--679,
  2021.

\bibitem{geelen2023grid}
Jim Geelen, O-joung Kwon, Rose McCarty, and Paul Wollan.
\newblock The grid theorem for vertex-minors.
\newblock {\em Journal of Combinatorial Theory, Series B}, 158:93--116, 2023.

\bibitem{gottesman1997stabilizer}
Daniel Gottesman.
\newblock {\em Stabilizer codes and quantum error correction}.
\newblock California Institute of Technology, 1997.

\bibitem{gottesman1998heisenberg}
Daniel Gottesman.
\newblock The {H}eisenberg representation of quantum computers.
\newblock In {\em Proc. XXII International Colloquium on Group Theoretical
  Methods in Physics, 1998}, pages 32--43, 1998.

\bibitem{gottesman2016surviving}
Daniel Gottesman.
\newblock Surviving as a quantum computer in a classical world.
\newblock {\em Textbook manuscript preprint}, 2016.

\bibitem{grier2022classification}
Daniel Grier and Luke Schaeffer.
\newblock The classification of clifford gates over qubits.
\newblock {\em Quantum}, 6:734, 2022.

\bibitem{hein2006entanglement}
Marc Hein, Wolfgang D{\"u}r, Jens Eisert, Robert Raussendorf, Maarten Van~den
  Nest, and H-J Briegel.
\newblock Entanglement in graph states and its applications.
\newblock In {\em Quantum computers, algorithms and chaos}, pages 115--218. IOS
  Press, 2006.

\bibitem{hlinveny2022twin}
Petr Hlin{\v{e}}n{\`y} and Filip Pokr{\`y}vka.
\newblock Twin-width and limits of tractability of {F}{O} model checking on
  geometric graphs.
\newblock {\em arXiv preprint arXiv:2204.13742}, 2022.

\bibitem{huang2020predicting}
Hsin-Yuan Huang, Richard Kueng, and John Preskill.
\newblock Predicting many properties of a quantum system from very few
  measurements.
\newblock {\em Nature Physics}, 16(10):1050--1057, 2020.

\bibitem{jena2024graph}
Andrew Jena.
\newblock Graph-theoretic techniques for optimizing {N}{I}{S}{Q} algorithms.
\newblock {\em PhD thesis, University of Waterloo}, 2024.

\bibitem{kim2023vertex}
Donggyu Kim and Sang-il Oum.
\newblock Vertex-minors of graphs: A survey.
\newblock {\em Discrete Applied Mathematics}, 351:54--73, 2024.

\bibitem{kumabe2024complexity}
Soh Kumabe, Ryuhei Mori, and Yusei Yoshimura.
\newblock Complexity of graph-state preparation by clifford circuits.
\newblock {\em arXiv preprint arXiv:2402.05874}, 2024.

\bibitem{lempel1975matrix}
Abraham Lempel.
\newblock Matrix factorization over {G}{F}$(2)$ and trace-orthogonal bases of
  {G}{F}$(2^n)$.
\newblock {\em SIAM Journal on Computing}, 4(2):175--186, 1975.

\bibitem{mccarty2021local}
Rose McCarty.
\newblock Local structure for vertex-minors.
\newblock {\em PhD thesis, University of Waterloo}, 2021.

\bibitem{mccarty2023decomposing}
Rose McCarty.
\newblock Decomposing a signed graph into rooted circuits.
\newblock {\em Advances in Combinatorics}, 2024:7, 22 pp., 2024.

\bibitem{mcclean2016theory}
Jarrod~R McClean, Jonathan Romero, Ryan Babbush, and Al{\'a}n Aspuru-Guzik.
\newblock The theory of variational hybrid quantum-classical algorithms.
\newblock {\em New Journal of Physics}, 18(2):023023, 2016.

\bibitem{nguyen2020average}
Huy-Tung Nguyen and Sang-il Oum.
\newblock The average cut-rank of graphs.
\newblock {\em European Journal of Combinatorics}, 90:103183, 2020.

\bibitem{nielsen2001quantum}
Michael~A Nielsen and Isaac~L Chuang.
\newblock {\em Quantum computation and quantum information}, volume~2.
\newblock Cambridge university press Cambridge, 2001.

\bibitem{patel2008optimal}
Ketan~N Patel, Igor~L Markov, and John~P Hayes.
\newblock Optimal synthesis of linear reversible circuits.
\newblock {\em Quantum Inf. Comput.}, 8(3):282--294, 2008.

\bibitem{raussendorf2001one}
Robert Raussendorf and Hans~J Briegel.
\newblock A one-way quantum computer.
\newblock {\em Physical review letters}, 86(22):5188, 2001.

\bibitem{schlingemann2001stabilizer}
Dirk Schlingemann.
\newblock Stabilizer codes can be realized as graph codes.
\newblock {\em Quantum Information and Computation}, 2(4):307--323, 2002.

\bibitem{shlosberg2023adaptive}
Ariel Shlosberg, Andrew~J Jena, Priyanka Mukhopadhyay, Jan~F Haase, Felix
  Leditzky, and Luca Dellantonio.
\newblock Adaptive estimation of quantum observables.
\newblock {\em Quantum}, 7:906, 2023.

\bibitem{van2004efficient}
Maarten Van~den Nest, Jeroen Dehaene, and Bart De~Moor.
\newblock Efficient algorithm to recognize the local clifford equivalence of
  graph states.
\newblock {\em Physical Review A—Atomic, Molecular, and Optical Physics},
  70(3):034302, 2004.

\bibitem{van2004graphical}
Maarten Van~den Nest, Jeroen Dehaene, and Bart De~Moor.
\newblock Graphical description of the action of local clifford transformations
  on graph states.
\newblock {\em Physical Review A}, 69(2):022316, 2004.

\bibitem{verteletskyi2020measurement}
Vladyslav Verteletskyi, Tzu-Ching Yen, and Artur~F Izmaylov.
\newblock Measurement optimization in the variational quantum eigensolver using
  a minimum clique cover.
\newblock {\em The Journal of chemical physics}, 152(12), 2020.

\end{thebibliography}

\appendix
\section{Graph States Ab Initio}\label{app:quantum}

\begin{definition}[Pauli Group \cite{gottesman2016surviving}]
The $n$-qubit \emph{Pauli group}, which we will denote by $\PauliGroup{2}[n]$, is a basis over which the $2^n \times 2^n$ Hermitian operators form a real vector space. Each element of $\PauliGroup{2}[n]$ is a tensor product of the following matrices together with an overall phase of $\pm 1$ or $\pm i$:
\begin{align*}
\pauli{I} = \begin{pmatrix} 1&0\\0&1 \end{pmatrix} \hspace{2em} \pauli{X} = \begin{pmatrix} 0&1\\1&0 \end{pmatrix} \hspace{2em} \pauli{Y} = \begin{pmatrix} 0&-i\\i&0 \end{pmatrix} \hspace{2em} \pauli{Z} = \begin{pmatrix} 1&0\\0&-1 \end{pmatrix}.
\end{align*}
\end{definition}

\begin{definition}[Clifford Group \cite{gottesman2016surviving}]
The $n$-qubit \emph{Clifford group}, which we will denote by $\CliffordGroup{2}[n]$ is the normalizer of the Pauli group. Formally,
\begin{align*}
\CliffordGroup{2}[n] = \{ g \text{ unitary} : gPg^\dagger \in \PauliGroup{2}[n],\ \forall P \in \PauliGroup{2}[n] \}.
\end{align*}
\end{definition}

Recall that we will be restricting our analysis to the $\{\H, \S, \CZ\}$ basis. Below are the definitions of the relevant Clifford operators which together span the group up to a global phase.
\begin{align*}
\H:\hspace{0.3em} &\frac{1}{\sqrt{2}}\begin{pmatrix} 1 & 1 \\ 1 & -1 \end{pmatrix} &
\S:\hspace{0em} &\begin{pmatrix} 1 & 0 \\ 0 & i \end{pmatrix} &
\CZ:\hspace{0em} &\begin{pmatrix} 1 & 0 & 0 & 0  \\ 0 & 1 & 0 & 0 \\ 0 & 0 & 1 & 0 \\ 0 & 0 & 0 & -1 \end{pmatrix}
\end{align*}
The above gates are called Hadamard, Phase, and Controlled-Z, respectively.

An $n$-qubit stabilizer state is a quantum state which is the $+1$ eigenvector or a set of exactly $2^n$ Paulis. Equivalently, a stabilizer state is a quantum state which can be obtained from $\ket{0}^{\otimes n}$ by applying \H{}, \S{}, and \CZ{} gates \cite{aaronson2004improved}. The graph state formalism provides a visual representation of $n$-qubit stabilizer states using $n$-vertex simple graphs.

\begin{definition}[Graph State \cite{hein2006entanglement}]
Given a simple graph, $G = (V,E)$, the corresponding \emph{graph state} is the stabilizer state,
\[
\ket{G} = \left( \prod_{(i,\ j) \in E} \CZ{}_{ij} \right) \ket{+}^{\otimes |V|}
\]
\end{definition}

The set of Paulis which stabilize a given graph state admits a simple basis which can be constructed directly from the adjacency matrix of the graph. The state, $\ket{G}$, corresponds to the basis
\[
\bigg\{ \pauli{X}_i \prod_{j \in N_G(i)} \pauli{Z}_j \bigg\}_{i=1}^{|V|}
\]

\end{document}